\documentclass[conference]{IEEEtran}
\IEEEoverridecommandlockouts
\usepackage{cite}
\usepackage{amsmath,amssymb,amsfonts}
\usepackage{graphicx}
\usepackage{textcomp}
\usepackage{xcolor}
\def\BibTeX{{\rm B\kern-.05em{\sc i\kern-.025em b}\kern-.08em
    T\kern-.1667em\lower.7ex\hbox{E}\kern-.125emX}}

\usepackage[english]{babel}
\usepackage[utf8]{inputenc}
\usepackage[T1]{fontenc}
\usepackage{lmodern}
\usepackage{braket}
\usepackage{wrapfig}
\usepackage{graphicx}
\usepackage{amsmath,amssymb,amsfonts,verbatim}
\usepackage{hyperref}
\usepackage{mathtools}
\usepackage{enumitem}

\usepackage{xr}

\usepackage{amsmath,amssymb,amsfonts,verbatim}
\usepackage{graphicx}
\usepackage[utf8]{inputenc}
\usepackage{textcomp}
\usepackage{xcolor}
\usepackage{braket}
\usepackage{soul}
\usepackage[colorinlistoftodos]{todonotes}

\usepackage{caption}
\usepackage{subcaption}

\usepackage{enumitem}

\usepackage{cite}
\usepackage{amsmath,amssymb,amsfonts,verbatim}
\usepackage{amsthm}
\usepackage{graphicx}
\usepackage[utf8]{inputenc}
\usepackage{textcomp}
\usepackage{xcolor}
\usepackage{braket}
\usepackage{soul}
\usepackage{enumitem}
\usepackage[colorinlistoftodos]{todonotes}

\usepackage{algpseudocode}
\usepackage{algorithm}

\usepackage{wrapfig}

\usepackage{caption}

\newtheorem{theorem}{Theorem}

\newtheorem{corollary}{Corollary}

\newtheorem{definition}{Definition}

\usepackage{tikz}
\usetikzlibrary{positioning}
 
\tikzstyle{pinstyle} = [pin edge={to-,thin,black}]
 
\allowdisplaybreaks

\newcommand{\nat}{\mathbb{N}}
\newcommand{\reals}{\mathbb{R}}

\newcommand{\dataset}{\boldsymbol{\beta}}

\newcommand{\gaussN}{\mathcal{N}}

\newcommand{\Prob}{\mathbb{P}}

\newcommand{\Tr}{\textrm{Tr}}

\newcommand{\barbeta}{\bar{\boldsymbol{\beta}}}
\newcommand{\optstat}{\Phi^*(\bar{\boldsymbol{\beta}})}
\newcommand{\optstati}{\hat{\Phi}^i(\barbeta)}
\newcommand{\stati}{\Phi^i}

\newcommand{\fcdf}{F_{\rvtest}}
\newcommand{\fccdf}{\bar{F}_{\rvtest}}
\newcommand{\Popt}{\mathbb{P}_{\optstat}}
\newcommand{\barstati}{\bar{\Phi}^i(\barbeta)}
\newcommand{\baroptstat}{\bar{\Phi}^*(\barbeta)}
\newcommand{\rvtest}{\Psi}
\newcommand{\rvtesti}{\Psi^i}

\newcommand{\numagents}{M}
\newcommand{\boldf}{\boldsymbol{f}}
\newcommand{\conset}{X_c^{\dtime}}
\newcommand{\effset}{X_E}
\newcommand{\simplex}{\mathcal{W}_{\numagents}}
\newcommand{\psimplex}{\simplex^+}
\newcommand{\woptset}{S(\sw)}

\newcommand{\lconst}{\alpha_{\dtime}}

\newcommand{\dtime}{t}
\newcommand{\respi}{\var^i_{\dtime}}
\newcommand{\nrespi}{\tilde{\var}^i_{\dtime}}
\newcommand{\noise}{\epsilon}
\newcommand{\noisei}{\noise^i}
\newcommand{\ndistr}{\Lambda}

\newcommand{\sw}{\mu}

\newcommand{\var}{\beta}
\newcommand{\varti}{\var_t^i}

\newcommand{\dimn}{N}

\newcommand{\numadv}{m}

\newcommand{\statei}{x}
\newcommand{\measi}{y^i}
\newcommand{\snoisei}{w}
\newcommand{\mnoisei}{v^i}
\newcommand{\stime}{t} 
\newcommand{\ftime}{k} 
\newcommand{\sdim}{q} 
\newcommand{\mdim}{p} 
\newcommand{\sncov}{Q_{\stime}(\lconst)}

\newcommand{\mncovi}{R_{\stime}(\varti)}
\newcommand{\argo}{\beta}

\newcommand{\kstate}{\hat{x}} 
\newcommand{\kcov}{\Sigma} 
\newcommand{\ARE}{\mathcal{A}(\lconst,\varti,\kcov)}

\begin{document}

\title{Statistical Detection of Coordination in a Cognitive Radar Network through Inverse Multi-objective Optimization\\
\thanks{This research was funded by National Science Foundation grant CCF-2112457,  Army Research office grant W911NF-21-1-0093 , and Air Force Office of Scientific Research grant FA9550-22-1-0016.}
}

\author{\IEEEauthorblockN{Luke Snow}
\IEEEauthorblockA{\textit{Electrical and Computer Engineering} \\
\textit{Cornell University}\\
Ithaca, NY \\
las474@cornell.edu}
\and
\IEEEauthorblockN{Vikram Krishnamurthy}
\IEEEauthorblockA{\textit{Electrical and Computer Engineering} \\
\textit{Cornell University}\\
Ithaca, NY \\
vikramk@cornell.edu}
}

\maketitle

\begin{abstract}
Consider a target being tracked by a cognitive radar network. If the target can intercept noisy radar emissions, how can it detect coordination in the radar network? By 'coordination' we mean that the radar emissions satisfy Pareto optimality with respect to multi-objective optimization over the objective functions of each radar and a constraint on total network power output. This paper provides a novel inverse multi-objective optimization approach for statistically detecting Pareto optimal ('coordinating') behavior, from a finite dataset of noisy radar emissions. Specifically, we develop necessary and sufficient conditions for radar network emissions to be consistent with multi-objective optimization (coordination), and we provide a statistical detector with theoretical guarantees for determining this consistency when radar emissions are observed in noise. We also provide numerical simulations which validate our approach. Note that while we make use of the specific framework of a radar network coordination problem, our results apply more generally to the field of inverse multi-objective optimization. 
\end{abstract}


\section{Introduction}
Cognitive radars \cite{haykin2012cognitive}, use the perception-action cycle of cognition to sense the target, learn relevant information, then optimally adapt their output emissions in response. We consider the case when there is a \textit{network} of cognitive radars which coordinate to optimally track a target. In a coordinating radar network, not only do the individual cognitive radars optimally adapt their output (with respect to an individual objective function) subject to resource constraints, but also the \textit{allocation of resources} between radars is subject to an optimization procedure. The resource to be allocated is often interpreted as the total power available to the radar network. Such an optimal power allocation strategy has been studied in (\cite{shi2017power}, \cite{panoui2014game}, \cite{chavali2011scheduling}, \cite{bacci2012game} and references therein), in which algorithmic game-theoretic methods are employed. Specifically, \cite{shi2017power} poses the problem of adaptive power allocation for radar networks as a cooperative game, and provides an iterative cooperative Nash bargaining algorithm which converges quickly to the Pareto optimal equilibrium.

However, in this work we are interested in the \textit{inverse} problem; namely, how can an external observer detect if a radar network is coordinating by observing its signals (in noise)? Can one then use these signals to reconstruct underlying objective functions which drive the network output, thus allowing for prediction of future responses? These questions, and extensions thereof, have been investigated in \cite{krishnamurthy2020identifying}, \cite{krishnamurthy2021adversarial}, \cite{pattanayak2022can}, \cite{pattanayak2022meta} in the framework of a \textit{single} cognitive radar. In this work we generalize these developments to a radar network (multi-objective optimization) framework. We define a coordinating radar network as a system which outputs signals which are Pareto efficient with respect to multi-objective optimization over each radar's objective and subject to a total power constraint; given observations of radar emissions we attempt to determine whether the radar network is coordinating, and subsequently reconstruct objective functions which closely match those in the multi-objective optimization. Thus, this problem is abstractly similar to inverse game theory \cite{kuleshov2015inverse} in that we aim to detect the output of a cooperative game (Pareto optimality), multi-agent inverse reinforcement learning \cite{natarajan2010multi} in that we aim to reconstruct feasible objective functions, and inverse multi-objective optimization \cite{dong2020expert}. 

Previous work \cite{snow2022identifying} has considered a similar problem of detecting coordination (multi-objective optimization) in a radar network based on \textit{deterministic} radar network signals. The key difference is that in this paper we assume the radar emissions  are \textit{observed in noise}. Specifically, the main contribution of this work is a statistical detector for identifying coordination from noisy observed signals. We provide theoretical guarantees on the probability of Type-I error of this detector, and demonstrate its efficacy via numerical simulations. The detector is based on a linear programming formulation, the feasibility of which is shown to be equivalent to the existence of a multi-objective optimization problem giving rise to the observed signals. 

We  emphasize that apart from radar networks, detecting multi-objective optimization by observing a black box applies to more general multi-agent {\em inverse} reinforcement learning in technological and social networks. 

This paper is organized as follows: In section~\ref{sec:MOO} we provide background on the problem of multi-objective optimization. In section~\ref{sec:system_model} we introduce the cognitive radar protocol and measurement model, and show how radar network coordination is equivalent to the multi-objective optimization framework presented in Section~\ref{sec:MOO}. In section~\ref{sec:MOO_detector} we provide necessary and sufficient conditions for the observed dataset of radar  emissions to be consistent with multi-objective optimization (Theorem~\ref{thm:cherchye1}), and we provide a statistical detector for determining whether the \textit{noisy} observed dataset is consistent with multi-objective optimization. Theoretical guarantees for this detector are given in Theorem~\ref{thm:stat_det}. 
In  section~\ref{sec:num_ex} we provide numerical studies that demonstrate the validity of our coordination-testing and objective reconstruction procedures. Finally we conclude in section~\ref{sec:conc}. 

\section{Background. Multi-Objective Optimization}
\label{sec:MOO}

Here we introduce the linearly constrained multi-objective optimization problem that is the basis of our problem formulation. We will  consider a cognitive radar network which distributes its power resources in such a way to solve this optimization, where each radar has a distinct objective function. For discrete time $\dtime$, subject to increasing and continuous linear function $\lconst \in \reals^n$ and optimization functions $f^i(\cdot): \reals^n \to \reals, i= [M] := \{1,\dots,\numagents\}$, the linearly constrained multi-objective optimization problem is given as: 
\begin{align}
\begin{split}
\label{eq:MOP}
 &\arg \max _\var \{f^1(\var),\dots,f^{\numagents}(\var)\}\ \\& s.t. \ \var \in \conset := \{\var \in \reals^n  : \lconst'\var \leq 1 \}
 \end{split}
\end{align}
where the linear constraint $\lconst \var$ is bounded by 1 without losing generality. In single-objective optimization, the goal is to find the best feasible argument which maximizes the objective. However, in multi-objective optimzation there will seldom exist an argument $\var$ which simultaneously maximizes all objectives, i.e. there will be tradeoffs between objectives for varying argument $\var$. Thus, the solution concept for the multi-objective optimization problem \eqref{eq:MOP} is that of \textit{efficiency}:
\begin{definition}{\textbf{Efficiency (Pareto Optimality)}}:
\label{def:par_opt}
 For fixed $\{\{f^i(\cdot)\}_{i=1}^{\numagents},\lconst\}$ and a vector $\var^*\in \conset$, let 
\begin{align*}
\begin{split}
    &Z^t(\var^*) = \{\var \in \conset : f^i(\var) \geq f^i(\var^*) \ \forall i \in [\numagents] \\
    &Y^t(\var^*) = \{\var \in \conset : \exists k : f^k(\var) > f^k(\var^*) \}
\end{split}
\end{align*}
The vector $\var^*$ is said to be \textit{efficient} if 
\begin{equation}
\label{efficiency}
Z^t(\var^*) \cap Y^t(\var^*) = \varnothing
\end{equation}
i.e., there does not exist another vector in the feasible set $\conset$ which increases the value of some objective $f^i(\cdot)$ without simultaneously decreasing the value of some other objective $f^j(\cdot)$, $i,j\in [\numagents]$.
\end{definition}
We then denote the set of all efficient solutions to the problem \eqref{eq:MOP} as 
\begin{equation}
\label{eq:effset}
\effset(\{f^i\}_{i=1}^M, \lconst) := \{\var^* \in \conset : \eqref{efficiency} \textrm{ is satisfied}\}
\end{equation}
and we say that $\var^*$ solves \eqref{eq:MOP} if and only if $\var^*$ is efficient, i.e.
\begin{align}
   &\var^* \in \{ \arg \max_\var \{f^1(\var),\dots,f^{\numagents}(\var)\}\ s.t. \ \var \in \conset \}\\
   &\Longleftrightarrow \ \var^* \in \effset(\{f^i\}_{i=1}^M, \lconst)
\end{align}

Denoting $\boldf(\var) = (f^1(\var),\dots,f^{\numagents}(\var))^T$, we can use the following problem of weighted sum (PWS) \cite{gass1955computational} to obtain an efficient solution:
\begin{equation}
\label{eq:PWS}
\max \ \sw^T\boldf(\var) \ s.t. \ \var \in \conset
\end{equation}
where $\sw = (\sw^1,\dots,\sw^{\numagents})^T \in \reals^{\numagents}_+$. The set of weights $\sw$ is restricted to the unit simplex, denoted as $\simplex := \{\sw \in \reals^{\numagents}_+ : \boldsymbol{1}^T\sw = 1\}$. Then we can denote the set of optimal solutions for \eqref{eq:PWS} as
\[\woptset = \arg \max_\var\{\sw^T\boldf(\var) : \var \in \conset\}
\]
Then, letting $\psimplex = \{\sw \in \reals^{\numagents}_{++} : \boldsymbol{1}^T\sw = 1\}$ denote the unit simplex with each weight $\sw^i$ strictly positive, we have \cite{miettinen2012nonlinear}:

\begin{equation}
\label{eq:effrel}
\bigcup_{\sw \in \psimplex}\woptset \subseteq \effset(\{f^i\}_{i=1}^M, \lconst) \subseteq \bigcup_{w \in \simplex} \woptset
\end{equation}

This relation will be useful for us in our result which states necessary and sufficient conditions for the radar network responses to be consistent with multi-objective optimization (coordination). We next present the radar network interaction model, and show how the multi-objective optimization framework presented here arises naturally.

\section{Radar Network Interaction Model}
\label{sec:system_model}

With the above background, we are now ready to discuss the cognitive radar network model.
We consider a radar network which optimally distributes its resources between $\numagents$ radars to track a target. The notion of 'optimally' coincides with Pareto optimality (Def. \ref{def:par_opt}). Specifically, at each time step the radar $i$ outputs signal $\varti$ such that the collective response $\{\varti\}_{i=1}^{\numagents}$ satisfies Pareto optimality with respect to each radar's objective and a joint power constraint. Abstractly, we take the point of view of the adversary which is being tracked by the radar network. We (the adversary) can interact with the network by performing purposeful maneuvers, and can observe noisy radar emissions in response to our maneuvers. Our aim is to determine whether or not the radar network is performing a multi-objective optimization (coordinating) to produce the emitted signals.  

\subsection{Interaction Dynamics}
Here we provide the general interaction dynamics between the cognitive radar and the adversary (us). We consider two time scales for the interaction: the fast time scale $\ftime = 1,2,\dots$ represents the scale at which the target state and measurement dynamics occur, and the slow time scale $\stime = 1,2,\dots$ represents the scale at which the probes $\lconst$ and radar responses $\{\varti\}_{i=1}^{\numagents}$ occur.  

\begin{definition}[Radar Network - Target Interaction]
The radar network - adversary interaction has the following dynamics:
\begin{align}
\label{inter_dynam}
    \begin{split}
        \textrm{target probe}: \lconst &\in \reals^{\dimn}_+ \\
        \textrm{radar i emission} : \varti &\in \reals^{\dimn}_+ \\
        \textrm{target state} : x_{\ftime} &\in \reals^{\sdim}, \
        x_{\ftime + 1} \sim p_{\lconst}(x| x_{\ftime}) \\
        \textrm{radar i observation}: \measi_{\ftime} &\in \reals^{\mdim}, \
        \measi_{\ftime} \sim p_{\varti}(y|x_{\ftime}) \\
        \textrm{radar i tracker}: \pi^i_{\ftime} &= \mathcal{T}(\pi^i_{\ftime-1},\measi_{\ftime}) 
    \end{split}
\end{align}
\end{definition}
where $\mathcal{T}$ represents a general Bayesian tracker. For a fixed $\stime$ in the slow time-scale, $\lconst$ abstractly represents a particular target maneuver (radial acceleration, etc.) which parametrizes the state update kernel, and $\varti$ abstractly represents radar $i$'s signal output which parametrizes its measurement kernel. These interaction dynamics are illustrated in Fig.~\ref{fig:blockdiag}. 
Taking the point of view of the target, we aim to detect if the radars are \textit{coordinating}:

\begin{definition}[Coordinating Cognitive Radar network]
\label{def:coord}
Considering the interaction dynamics \eqref{inter_dynam}, we define a coordinating cognitive radar network to be a network of $\numagents$ radars , each with individual monotone increasing objective functions $f^i: \reals^{\dimn} \to \reals, i\in[\numagents]$, which produces output signals $\{\varti\}_{i=1}^{\numagents}$ on the slow time-scale in accordance with
\begin{align}
\label{def:coord_eq}
    \begin{split}
        &\{\varti\}_{i=1}^{\numagents} \in \arg\max_{\{\argo^i\}_{i=1}^{\numagents}} \{f^1(\argo^1),\dots,f^{\numagents}(\argo^{\numagents})\} \\
        & s.t. \quad \lconst (\sum_{i=1}^{\numagents} \argo^i) \leq 1 
    \end{split}
\end{align}
    
\end{definition}

Note that \eqref{def:coord_eq} is a special case of the general problem in \eqref{eq:MOP}. Thus, a coordinating cognitive radar network emits signals which are \textit{efficient (Pareto optimal)} (Def. \eqref{def:par_opt}) in order to optimally parametrize the measurement kernels (through e.g., increasing measured signal power) subject to each objective function, the state dynamics of the target, and a constraint on the total power output.

Next, we specify a particular concrete example of these interaction dynamics in which the spectra of state and measurement noise covariance matrices act as the probe and response. We justify how this provides a natural interpretation of the above abstract framework.  \\
\textit{Remark}: A multi-target interaction can be incorporated into the above framework by considering $\lconst$ to be the vector of state-noise spectral norms of each target. We exclude this for brevity, but consider it in future work.
\subsection{Constrained Spectral Optimization}
Linear Gaussian dynamics for a target's kinematics \cite{li2003survey} and linear Gaussian measurements at each radar are widely assumed as a useful approximation \cite{bar2004estimation}. Thus we will consider the following linear Gaussian state dynamics and measurements over the \textit{fast time scale $\ftime \in \nat$}:
\begin{align}
    \begin{split}
    \label{lin_gaus}
        \statei_{\ftime+1} &= A\statei_{\ftime} + \snoisei_{\ftime}, \  \statei_0 \sim \pi^i_0, \\
        \measi_{\ftime} &= C^i\statei_{\ftime} + \mnoisei_{\ftime}, \ i\in[\numadv]
    \end{split}
\end{align}
where $\statei_{\ftime}, \snoisei_{\ftime} \in \reals^{\sdim}$ are the target state and noise vectors, respectively, and $A \in \reals^{\sdim\times \sdim}$ is the state update matrix. $\measi_n \in \reals^{\mdim}$ is the $i$'th radar measurement of the target, $C^i \in \reals^{\mdim \times \sdim}$ is the $i$'th radar measurement transformation, and $\mnoisei_n \in \reals^{\mdim}$ is the measurement noise. The constraints and subsequent radar responses will be indexed over the \textit{slow time scale} $\stime \in \nat$. Abstractly, these will parameterize the state and noise covariance matrices:
\begin{equation}
    \snoisei_{\ftime} \sim \gaussN(0,\sncov), \ \mnoisei_{\ftime} \sim \gaussN(0,\mncovi)
\end{equation}

\begin{figure}
\centering
  \includegraphics[width=\linewidth,scale=0.2]{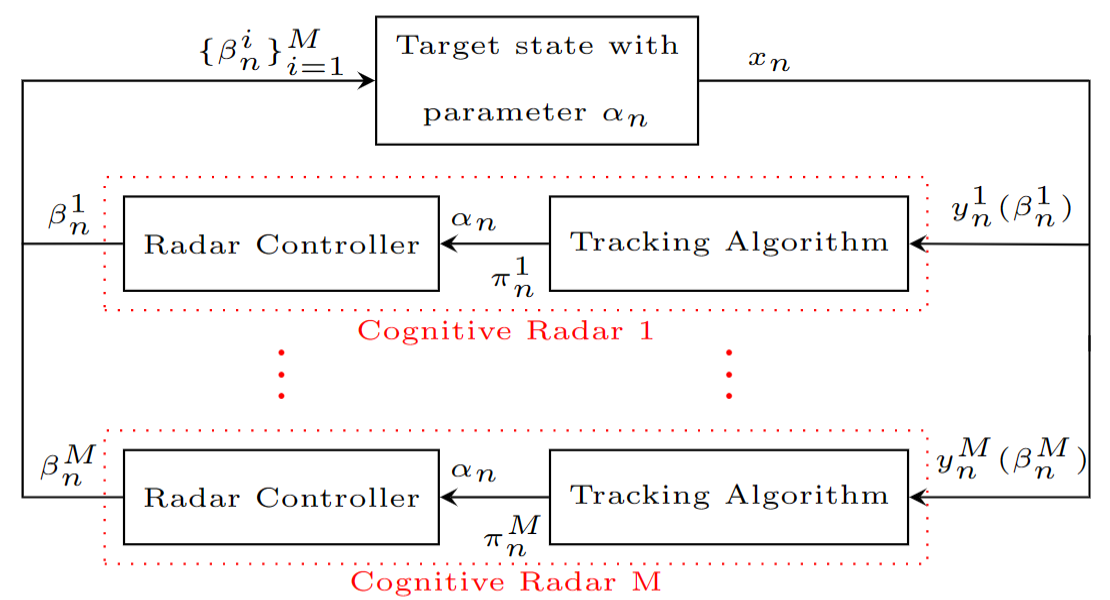}
  \caption{\scriptsize Interaction of our dynamics with the adversary's cognitive radar network. Each cognitive radar is comprised of a Bayesian tracker and a radar controller. Based on the time series $\{\varti\}_{i=1}^{\numagents}, t=1,\dots,T$, our goal is to determine if the radar network is coordinating, i.e., if Def.~\ref{def:coord} is satisfied.}
  \label{fig:blockdiag}
\end{figure}

 In this spectral interpretation, $\lconst$ represents the vector of eigenvalues of state-noise covariance matrix $Q_t$ and $\varti$ represents the vector of eigenvalues of the inverse measurement noise covariance matrix $R_t^{-1}$. 
The radar network tracks our target using Kalman filter trackers: \\
Based on observations $\measi_1,\dots,\measi_{\ftime}$ of the target, the tracking functionality in the $i$'th radar computes the target state posterior
\begin{equation*}
\label{eq:kalpost}
    \pi_{\ftime}^i = \gaussN(\kstate_{\ftime}^i,\kcov_{\ftime}^i)
\end{equation*}
where $\kstate_{\ftime}^i$ is the conditional mean state estimate and $\kcov_{\ftime}^i$ is the covariance, computed by the classical Kalman filter:
\begin{align*}
    \begin{split}
        \kcov_{\ftime + 1 | \ftime}^i &= A\kcov_{\ftime}^iA' + \sncov \\
        K_{\ftime + 1}^i &= C^i\kcov_{\ftime+1 | \ftime}^i (C^i)' + \mncovi \\
        \kstate_{\ftime+1}^i &= A\kstate^i + \kcov_{\ftime+1|\ftime}^i(C^i)'(K^i_{\ftime+1|\ftime})^{-1}(\measi_{\ftime+1} - C^iA\kstate^i_{\ftime}) \\
        \kcov^i_{\ftime+1} &= \kcov^i_{\ftime+1|\ftime} - \kcov^i_{\ftime+1|\ftime}(C^i)'(K_{\ftime+1}^i )^{-1}C^i\kcov^i_{\ftime+1|\ftime}
    \end{split}
\end{align*}
Under the assumption that the model parameters in \eqref{lin_gaus} satisfy $[A,C^i]$ is detectable and $[A,\sqrt{\sncov}]$ is stabilizable, the asymptotic predicted covariance $\kcov^i_{\ftime+1 | \ftime}$ as $k \to \infty$ is the unique non-negative definite solution of the \textit{algebraic Riccatti equation} (ARE): 
\begin{align*}
    \begin{split}
    \label{eq:ARE}
        &\ARE := \\
        &- \kcov + A(\kcov - \kcov (C^i)'[C^i\kcov (C^i)' + \mncovi]^{-1}C^i\kcov)A' \\& \quad + \sncov = 0
    \end{split}
\end{align*}
 Suppose each radar aims to optimize its unique objective function $f^i$ in isolation. Let $\kcov_{\stime}^{* -1}(\lconst,\varti)$ denote the solution of the ARE. Also suppose that the radar can only expend sufficient resources to ensure that the precision (inverse covariance) is at most some pre-specified precision $\bar{\kcov}^{-1}$. The radar would then adaptively choose the best waveform, corresponding one-to-one with the measurement noise covariance spectrum $\varti$, to meet the objective while satisfying this resource constraint, i.e.
  \begin{equation}
  \label{eq:nl_con_opt}
     \respi \in \arg\max_{\var} f^i(\var) \ : \ \kcov_{\stime}^{* -1}(\lconst,\varti) \leq \bar{\kcov}^{-1}
 \end{equation}
 and by Lemma 3 of \cite{krishnamurthy2020identifying} we can recover a linear constraint from this formulation, i.e.
 \[\kcov_{\stime}^{* -1}(\lconst,\varti) \leq \bar{\kcov}^{-1} \Rightarrow \lconst'\varti \leq 1\] The key idea behind this equivalence is to show the asymptotic precision $\kcov^*_n(\lconst,\cdot)$ is monotone increasing in the second argument $\varti$ using the information Kalman filter formulation. Thus, we can abstract \eqref{eq:nl_con_opt} to the following optimization with linear constraint
 \begin{equation}
     \respi \in \arg\max_{\var} f^i(\var) \ : \ \lconst'\var \leq 1
 \end{equation}

 \subsection{Multi-Objective Spectral Optimization}
Now we consider the case when the radar network is jointly constrained by a total power bound.  Since increased power output for the $i$'th radar signal corresponds directly to increased measurement $i$ precision, we can abstract the joint power constraint among all $\numagents$ radars to 
$\lconst' (\sum_{i=1}^{\numagents}\varti ) \leq p^*$
where $p^*$ is the constraint on total network power output. In this case, the cognitive radar network optimization problem becomes \eqref{def:coord_eq}.

Let us make the assumption that $\varti > 0 \ \forall t \in [T], i \in [\numagents]$, i.e., each radar always outputs a non-zero power signal. In the Appendix we prove a technical Lemma which allows us to make the following equivalence: \eqref{eq:effrel} together with \eqref{lemeq:mu0} implies that there exists $\sw \in \psimplex$ such that the expression \eqref{def:coord_eq} is equivalent to 
\begin{align}
\begin{split}
\label{eq:moo}
    \{\varti\}_{i=1}^{\numagents} \in &\arg\max_{\{\argo^i\}_{i=1}^{\numagents}} \sum_{i=1}^{\numagents} \sw^i f^i(\argo^i) \\ 
     &s.t. \quad \lconst' (\sum_{i=1}^{\numagents}\argo^i ) \leq p^*
\end{split}
\end{align}

Recall that we are interested in the inverse multi-objective optimization problem; in the following section we provide a necessary and sufficient condition for the existence of objective functions for which the observed signals $\{\varti\}_{i=1}^{\numagents}$ satisfy constrained multi-objective optimization.
\section{Detection of Coordination} 
\label{sec:MOO_detector}

First we provide the equivalence of cognitive radar network coordination (Def.~\ref{def:coord}) to a linear program formulation, and a subsequent objective function reconstruction equation. We then utilize this in a statistical detector for determining whether noisy network responses are consistent with multi-objective optimization (coordination).Finally we provide an algorithm for implementing this detector and objective function reconstruction. We assume the target can observe the signals $\{\varti, t\in[T]\}_{i=1}^{\numagents}$ through e.g., an omni-directional receiver.

\subsection{Equivalence to Linear Program}
Suppose we have the dataset of constraints and system responses $\dataset = \{\lconst, \{\respi\}_{i=1}^{\numagents}, t \in [T] \}$. Here we provide a necessary and sufficient condition for the dataset $\dataset$ to be consistent with multi-objective optimization.
\begin{theorem}
 \label{thm:cherchye1}
    Let $\dataset$ be a set of observations. The following are equivalent:
    \begin{enumerate}
    \item there exist a set of $M$ concave and continuous objective functions $U^1,\dots,U^m$, weights $\sw \in \psimplex$ and constraint $p^*$ such that $\forall t \in [T]$:
    \begin{align}
    \begin{split}
    \label{thm1:rat}
        \{\respi\}_{i=1}^{\numagents} \in &\arg\max_{\{\argo^i\}_{i=1}^{\numagents}} \sum_{i=1}^{\numagents} \sw^i U^i(\argo^i) \\
        &s.t. \quad \lconst' (\sum_{i=1}^{\numagents}\argo^i ) \leq p^*
    \end{split}
    \end{align}
    \item there exist numbers $u_j^i > 0, \lambda_j^i > 0$ such that for all $s,t \in [T]$, $i \in [M]$: 
    \begin{equation}
    \label{af_ineq}
        u_s^i - u_t^i - \lambda_t^i\lconst'[\var_s^i - \varti] \leq 0
    \end{equation}
    \end{enumerate}
\end{theorem}
\begin{proof}
See Theorem 1 of \cite{snow2022identifying}
\end{proof}
This allows us to simply solve the linear program feasibility test \eqref{af_ineq} to test for multi-objective optimization. Specifically, given the equivalence of \eqref{eq:moo} and \eqref{def:coord_eq}, we can use this linear programming formulation to test for coordination in the cognitive radar network. 

\begin{corollary}
\label{cor:Utrec}
Given constants $u_t^i, \lambda_t^i, t\in[T],i\in[M]$ which make \eqref{af_ineq} feasible, explicit monotone and continuous objective functions that "rationalize" the dataset $\{\lconst, \respi, t \in [T], i \in [M]\}$ are given by
\begin{equation}
\label{eq:Utrec}
     U^i(\cdot) = \min_{t \in [T]} \left[u_t^i + \lambda_t^i\lconst'[\cdot - \respi] \right]
\end{equation}
i.e., \eqref{thm1:rat} is satisfied.
\end{corollary}
\begin{proof}
See Lemma 1 of \cite{snow2022identifying}.
\end{proof}

This Corollary provides us with a mechanism for reconstructing objective functions which rationalize the observed responses, giving us a way to predict future cognitive radar network outputs. 

Recall that thus far we have considered only deterministic radar $i$ responses $\varti$. We now consider the case when these measured responses are corrupted by noise. We next provide a statistical detector for determining whether these \textit{noisy} responses are consistent with multi-objective optimization, with theoretical guarantees on Type-I error. We then provide a general scheme for reconstructing objective functions which most closely rationalize the observed noisy responses. 

\subsection{Statistical Detector}
\label{sec:StatDet}

 Let $\barbeta$ denote the dataset when the radar responses are  observed in noise:
\begin{equation}
\label{eq:dataset}
    \barbeta = \{\lconst, \nrespi , t \in [T], i \in [\numagents]\}
\end{equation}
where $\nrespi = \respi + \noisei_t$, and $\noise^i_t$ are i.i.d. and distributed according to some distribution $\ndistr^i_t$. 
We propose a statistical detector to optimally determine if the responses are consistent with Pareto optimality \eqref{eq:MOP}. Define \\
$H_0$: null hypothesis that the dataset \eqref{eq:dataset} arises from the optimization problem \eqref{def:coord_eq}. \\
$H_1$: alternative hypothesis that the dataset \eqref{eq:dataset} does not arise from the optimization problem \eqref{def:coord_eq}. 

There are two possible sources of error: \\
\textbf{Type-I error}: Reject $H_0$ when $H_0$ is valid.\\
\textbf{Type-II error}: Accept $H_0$ when $H_0$ is invalid. 

We formulate the following test statistic $\optstat$, as a function of $\barbeta$, to be used in the detector: 
\begin{equation}
\label{eq:optstat}
\optstat = \max_i \optstati
\end{equation}
where $\optstati$ is the solution to:
\begin{align}
\begin{split}
\label{eq:LP}
&\min \stati : \exists \ u_t^i > 0, \lambda_t^i > 0 : \\
&u_s^i - u_t^i - \lambda_t^i \lconst'(\bar{\var}^i_s - \bar{\var}_t^i) - \lambda_t^i \stati \leq 0 
\end{split}
\end{align}
Form the random variable $\rvtest$ as 
\begin{align}
\begin{split}
\label{eq:psimax}
&\rvtest = \max_i \rvtesti \\
&\rvtesti = \max_{t\neq s}[\lconst'(\noisei_t - \noisei_s)]
\end{split}
\end{align}
Then we propose the following statistical detector (with $\gamma \in (0,1)$):
\begin{equation}
\label{eq:stat_test}
\int_{\optstat}^{\infty} f_{\rvtest}(\psi)d\psi 
\begin{cases}
\geq \gamma \Rightarrow H_0 \\ < \gamma \Rightarrow H_1
\end{cases}
\end{equation}
where $f_{\rvtest}(\cdot)$ is the probability density function of $\rvtest$. Let $\fcdf$ be the cdf of $\rvtest$ and $\fccdf$ be the complementary cdf of $\rvtest$. Then we have the following guarantees:

\begin{theorem}
\label{thm:stat_det}
Consider the noisy dataset $\eqref{eq:dataset}$, and suppose \eqref{eq:LP} has a feasible solution. Then 
\begin{enumerate}
    \item The following null hypothesis equivalence holds:
    \begin{equation}
    \label{eq:H0equiv} H_0 \Longleftrightarrow \bigcap_{i \in [M]} \{\optstati \leq \rvtesti\} \end{equation}
    \item The probability of Type-I error (false alarm) is 
    \[ \Popt(H_1 | H_0) := \Prob(\fccdf(\optstat) \leq \gamma \ | H_0) \leq \gamma \]

    \item The optimizer $\optstat$ yields the smallest Type-I error bound:
    \begin{align*}
    \begin{split}
    &\Prob_{\bar{\boldsymbol{\Phi}}(\barbeta)}(H_1 | H_0) \geq \Popt(H_1 | H_0) \quad \\
    &\quad \forall \bar{\boldsymbol{\Phi}}(\barbeta) \in [\optstat, \rvtest]
    \end{split}
    \end{align*}
    
\end{enumerate}
\end{theorem}

\begin{proof}
    See Appendix \ref{pf:stat_det}
\end{proof}

The contribution of this detector is that it provides a strict guarantee on the upper bound of probability of Type-I error; the specific choice of threshold $\gamma$ is left to any particular problem application and may vary depending on design criteria. 
\subsection{Statistical Detector Implementation and objective Reconstruction}

Here we present an implementable algorithm for detecting coordination in the radar network and reconstructing objective functions which most closely rationalize the observed noisy responses. 

In practice we would likely not have access to the density function $f_{\Psi}(\cdot)$. However, we would likely have some assumptions on the noise statistics captured by the distributions $\Lambda_t^i$, such as additive Gaussian noise. Thus, we propose to compute an approximation $\hat{F}_{\Psi}(\cdot)$ of the cumulative distribution function $F_{\Psi}(\cdot)$ using assumptions on the noise statistics, then implement the statistical detector through this. Algorithm 1 provides a practically feasible implementation of the statistical detector \eqref{eq:stat_test}.

Recall that Corollary \ref{cor:Utrec} gives us the ability to reconstruct objective functions for which the observed \textit{deterministic} responses $\{\varti\}_{i=1}^{\numagents}$ are consistent with multi-objective optimization. If the statistical detector suggests that the radar network is coordinating, i.e. $H_0$ holds, it would be in our interest to obtain these objective functions. We can do so by utilizing the parameters $\{u_t^i, \lambda_t^i, t\in[T],i\in[\numagents]\}$ which solve \eqref{eq:LP}. Note that due to the additive noise in the measured signals $\{\hat{\var}_t^i\}_{i=1}^{\numagents}$, Corollary \ref{cor:Utrec} does not guarantee that the signals can \textit{exactly} be rationalized by these reconstructed functions. However, using the parameters which solve \eqref{eq:LP} is a useful heuristic, and we demonstrate this validity of this reconstruction in a numerical example. 


\begin{algorithm}[t]
\label{alg:MOOdet}
\caption{Detecting Multi-Objective Optimization}
\begin{algorithmic}[1]
    \For{l=1:L}
        \For{i=1:M}
             \State simulate $\boldsymbol{\noise}^i_l = [\noise^i_1,\dots,\noise^i_N]^{(l)}, \quad \noise^i_t \sim \Lambda_t^i$
        \EndFor
        \State Compute $\Psi^{l} := \max_i \{\max_{t\neq s}[\lconst(\noise_t^i - \noise_s^i)]\}$
    \EndFor
    \State Compute $\hat{F}_{\Psi}(\cdot)$ from $\{\Psi^l\}_{l=1}^L$
\State Record radar network response $\barbeta$ to the probe $\lconst$
\State Solve \eqref{eq:optstat} for $\optstat$
\State Save $\mathcal{P} := \{\hat{u}_t^i, \hat{\lambda}_t^i, t \in [T], i \in [\numagents]\}$ such that
\begin{align*}
    \hat{u}_s^i - \hat{u}_t^i - \hat{\lambda}_t^i \lconst'(\bar{\var}^i_s - \bar{\var}_t^i) - \hat{\lambda}_t^i \optstati \leq 0 \ \forall i \in [\numagents]
\end{align*}
\State Implement detector \eqref{eq:stat_test} as
\begin{equation}
    1  - \hat{F}_{\Psi}(\optstat) \begin{cases}
         > \gamma \Rightarrow H_0 \\
         \leq \gamma \Rightarrow H_1
    \end{cases}
\end{equation}
\If {$H_0$}
 \State Reconstruct objective functions from \eqref{eq:NumRec}
\EndIf
\end{algorithmic}
\end{algorithm}


\section{Numerical Studies}
\label{sec:num_ex}

For our numerical examples we consider the case with $M=3$ radars, outputting signals $\varti \in \reals^2$, with objective functions given by 
\begin{align}
\begin{split}
\label{eq:NumUtils}
    &f^1(\var) = \textrm{det}(R^{-1}(\beta)) = \beta(1)\times\beta(2), \\
    &f^2(\beta) = \Tr(R^{-1}(\beta)) = 
    \beta(1)+\beta(2), \\
    &f^3(\beta) = \sqrt{\beta(1)}\beta(2)
\end{split}
\end{align}

\subsection{Statistical Detector Performance}

\begin{figure}
\centering
  \includegraphics[width=\linewidth,scale=0.2]{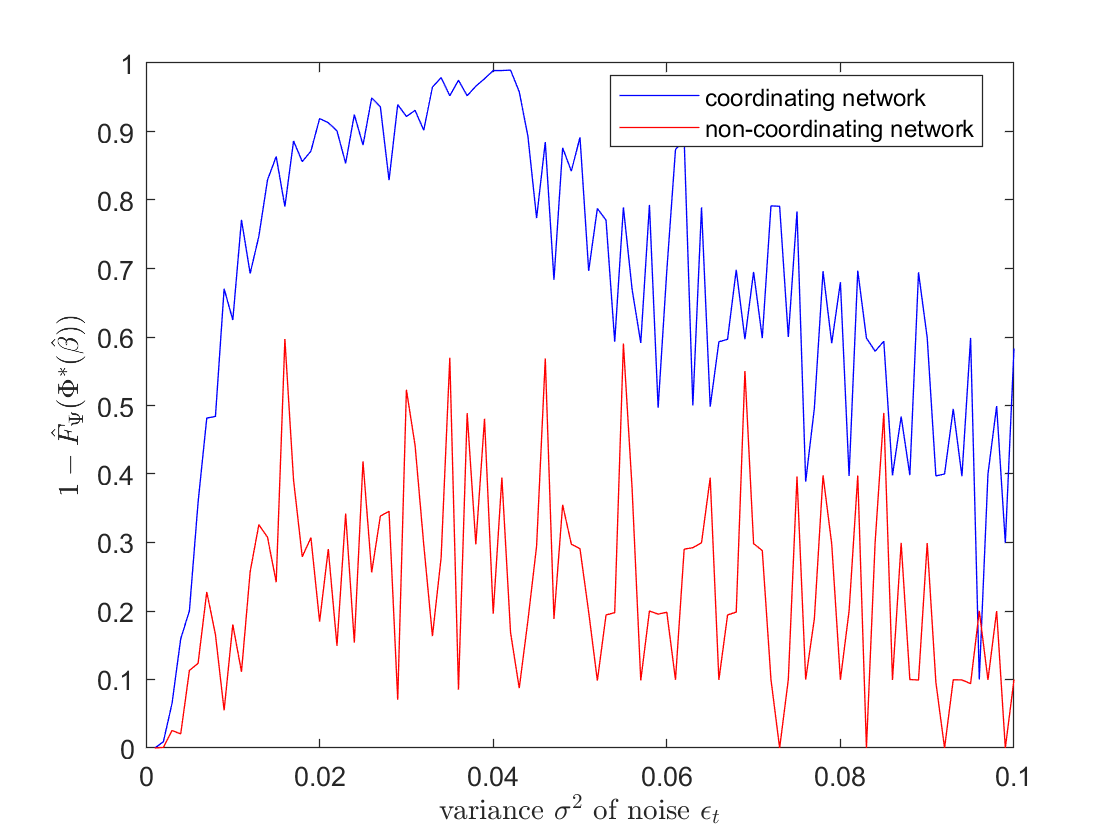}
  \caption{\small Statistic $1  - \hat{F}_{\Psi}(\optstat)$ as a function of variance of the noise distribution $\Lambda_t$. Higher $1  - \hat{F}_{\Psi}(\optstat)$ corresponds to higher likelihood of radar network coordination in the statistical detector \eqref{eq:stat_test}.}
  \label{fig:StatDet}
\end{figure}

Here we investigate the empirical behavior of the statistic $1  - \hat{F}_{\Psi}(\optstat)$ under both $H_0$ and $H_1$. We generate the statistic from the procedure outlined in Algorithm 1, with $L = 500$, $M = 3$, $T=10$. The probe signal $\lconst \in \reals^2$ is generated randomly as $\lconst \sim U[0.1,1.1]^2 $, i.e. each element of $\lconst$ is generated as an independent uniform random variable on the interval [0.1,1.1]. To simulate a cognitive radar network, the responses $\{\varti\}_{i=1}^{\numagents}$ are taken as solutions to the multi-objective optimization \eqref{def:coord} with objective functions given by \eqref{eq:NumUtils}, and $\sw^1=\sw^2=0.4,\sw^3=0.3$. Then noisy responses $\{\bar{\var}_t^i\}_{i=1}^{\numagents}$ are obtained by adding i.i.d. Gaussian noise $\noisei_t \sim \Lambda_t = \gaussN(0,\sigma^2)$. The blue line in Figure~\ref{fig:StatDet} displays the resultant empirical statistic $1  - \hat{F}_{\Psi}(\optstat)$ as a function of noise variance. 
To simulate a non-coordinating radar network, we generate each response $\varti \sim U[0,1]^2$ independently, and similarly add Gaussian measurement noise $\noisei_t \sim \Lambda_t = \gaussN(0,\sigma^2)$. The red line in Figure~\ref{fig:StatDet} is the empirical statistic $1  - \hat{F}_{\Psi}(\optstat)$ under these circumstances, when no coordination is present.

Let us interpret the simulation results displayed in Figure~\ref{fig:StatDet}. Observe that the statistic $1  - \hat{F}_{\Psi}(\optstat)$ is consistenly larger when the radar network coordinates. This validates our choice that the null hypothesis $H_0$ (coordination) should be chosen once the statistic surpasses a threshold. Under $H_0$ (coordination), the statistic begins to decrease as the noise variance increases. This intuitively should hold, since as the noise increases the signal structure imposed by the multi-objective optimization begins to degenerate. Also observe that the statistic goes to zero as the noise variance goes to zero. This is somewhat counter-intutitive, as one might think that in the deterministic limit (no noise) the detector should always be able to identify coordination. However, notice that as the noise variance goes to zero the cumulative distribution function $F_{\Psi}(\cdot)$ will resemble a unit step function, and, by the simulation, seems to  do so faster than the statistic $\optstat$ converges to zero. It would be interesting to investigate this phenomenon further. 
\subsection{Reconstructing objective Functions}

Figure~\ref{fig:utilities} displays the three optimized objective functions given in \eqref{eq:NumUtils} in the left column, and the three reconstructed objective functions $U^1(\cdot),U^2(\cdot),U^3(\cdot)$ given by
\begin{equation}
\label{eq:NumRec}
U^i(\cdot) = \min_{t \in [T]} \left[\hat{u}_t^i + \hat{\lambda}_t^i\lconst'[\cdot - \respi] \right]
\end{equation}
where $\{\hat{u}_t^i, \hat{\lambda}_t^i, t \in [T], i \in [\numagents]\}$ is taken from $\mathcal{P}$ in Algorithm 1. We note again that these reconstructed objective functions may not exactly rationalize the responses due to the additive noise, but as can be seen in Fig.~\ref{fig:utilities}, the heuristic \eqref{eq:NumRec} succeeds in approximating the true objective functions. An interesting future line of work is to theoretically analyze how well these reconstructed objective functions approximate the true objective functions, taking into account the noise statistics.

\begin{figure}[t]
        \centering
        \begin{subfigure}[]{0.23\textwidth} 
            \centering
            \includegraphics[width=\textwidth]{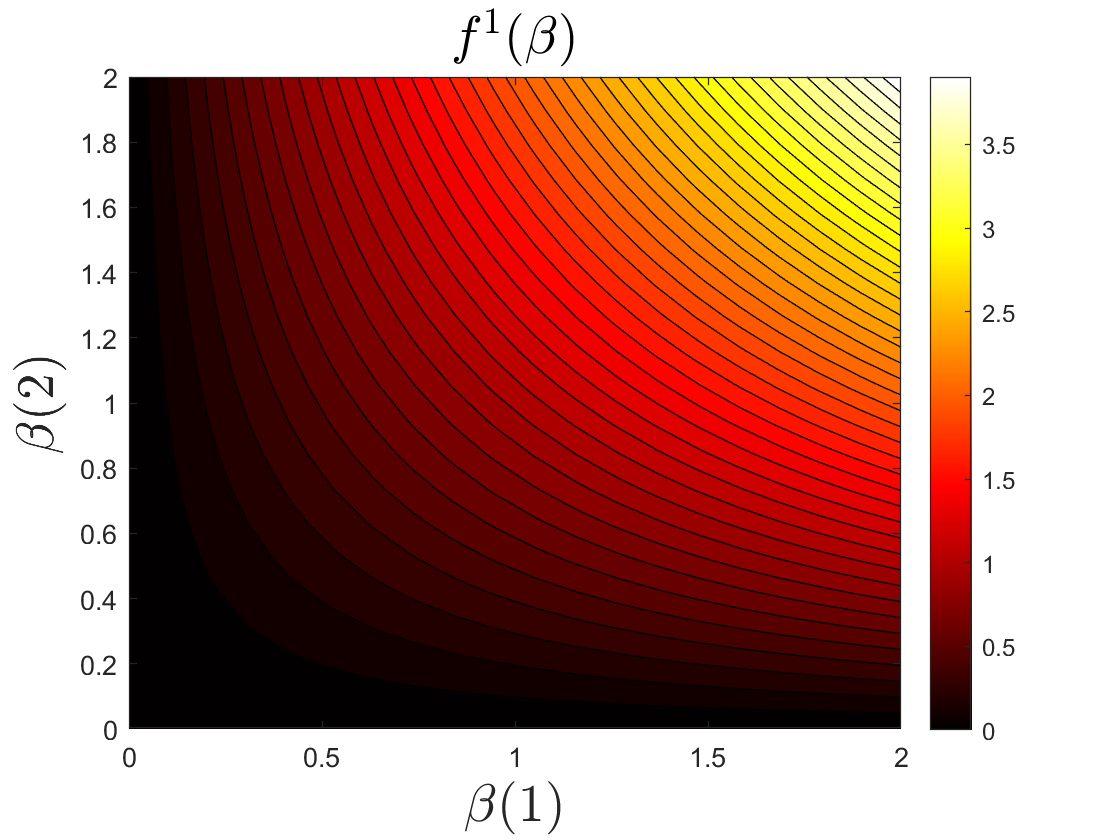}
            \caption[]%
            {{\scriptsize $f^1(\beta) = \textrm{det}(R^{-1}(\beta))$}}    
            \label{fig:T_util1}
        \end{subfigure}
        \hfill
        \begin{subfigure}[]{0.23\textwidth}  
            \centering 
            \includegraphics[width=\textwidth]{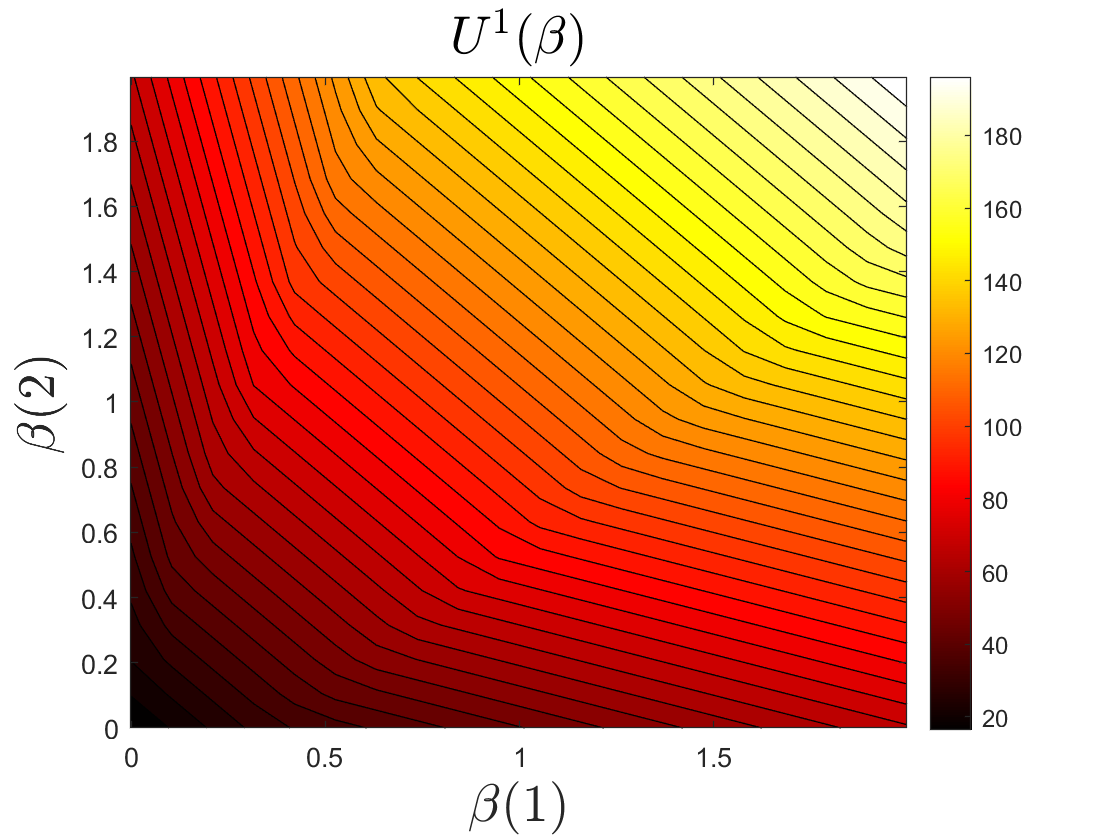}
            \caption[]%
            {{\scriptsize $U^1(\beta)$}}    
            \label{fig:util1}
        \end{subfigure}
        \vskip\baselineskip

        \begin{subfigure}[]{0.23\textwidth}   
            \centering 
            \includegraphics[width=\textwidth]{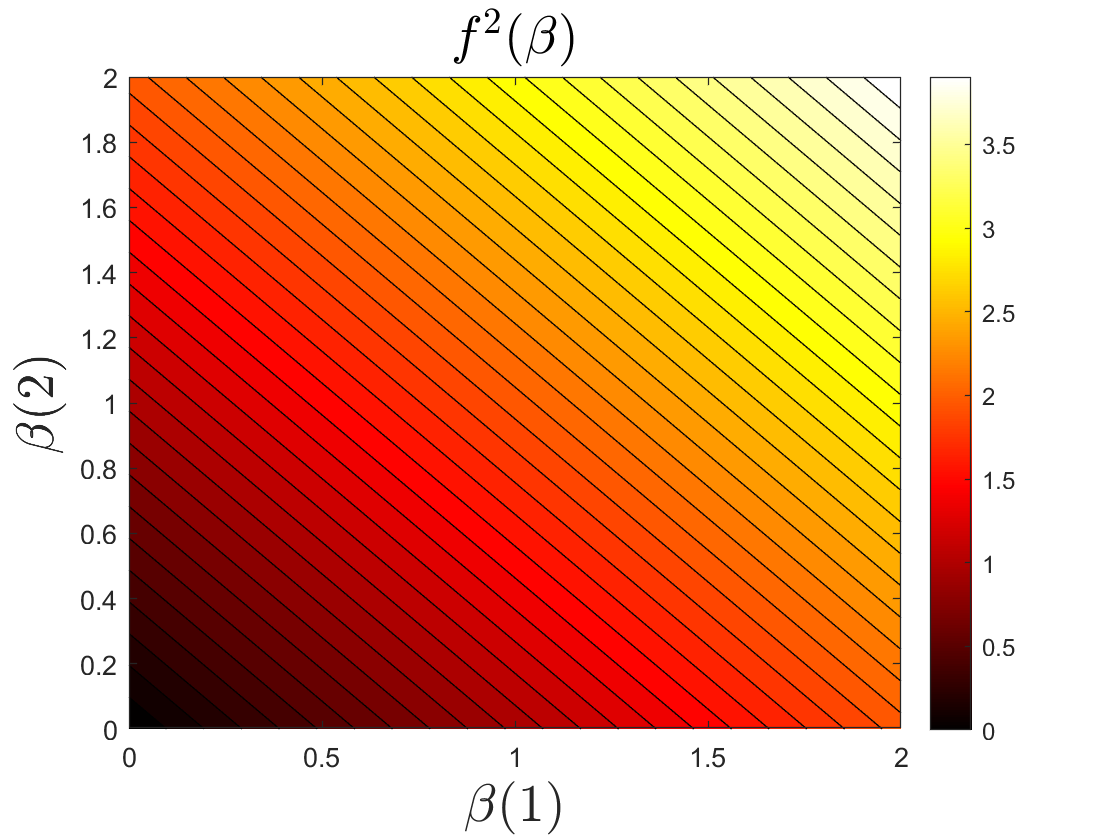}
            \caption[]%
            {{\scriptsize $f^2(\beta) = \textrm{Tr}(R^{-1}(\beta))$}}    
            \label{fig:T_util2}
        \end{subfigure}
        \hfill
        \begin{subfigure}[]{0.23\textwidth}   
            \centering 
            \includegraphics[width=\textwidth]{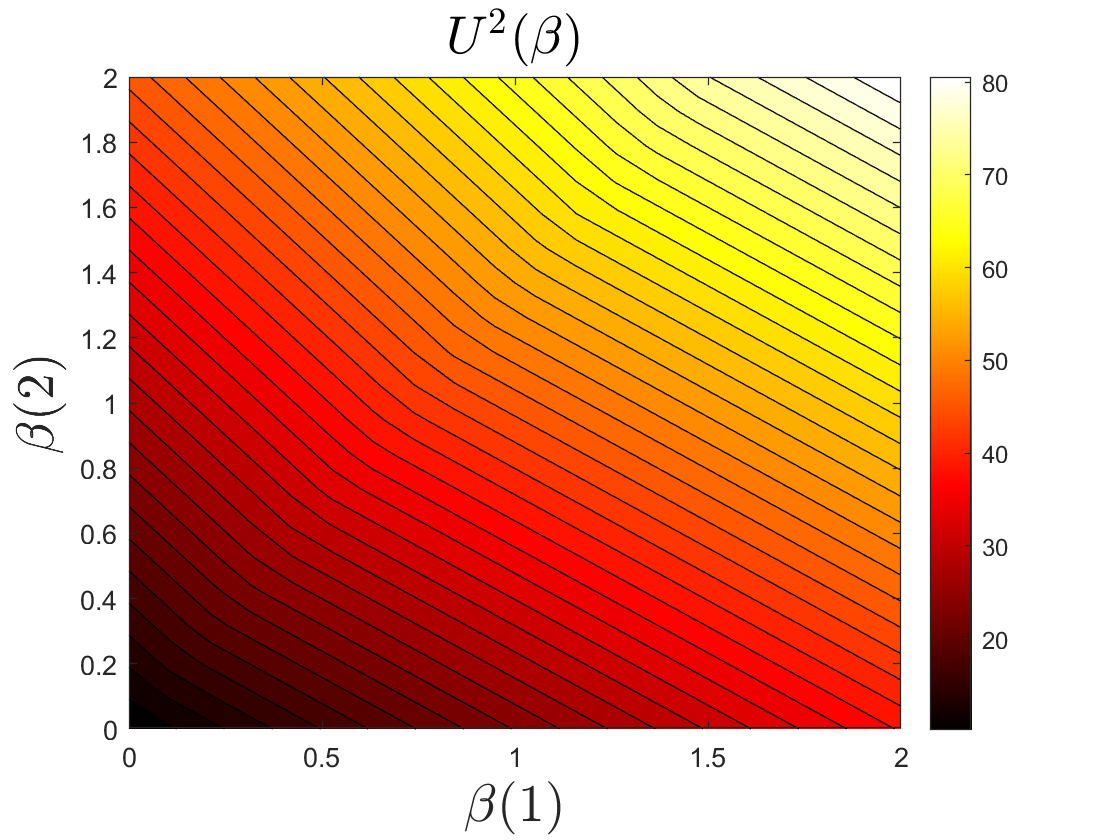}
            \caption[]%
            {{\scriptsize $U^2(\beta)$}}    
            \label{fig:util2}
        \end{subfigure}
        \vskip\baselineskip

        \begin{subfigure}[]{0.23\textwidth}   
            \centering 
            \includegraphics[width=\textwidth]{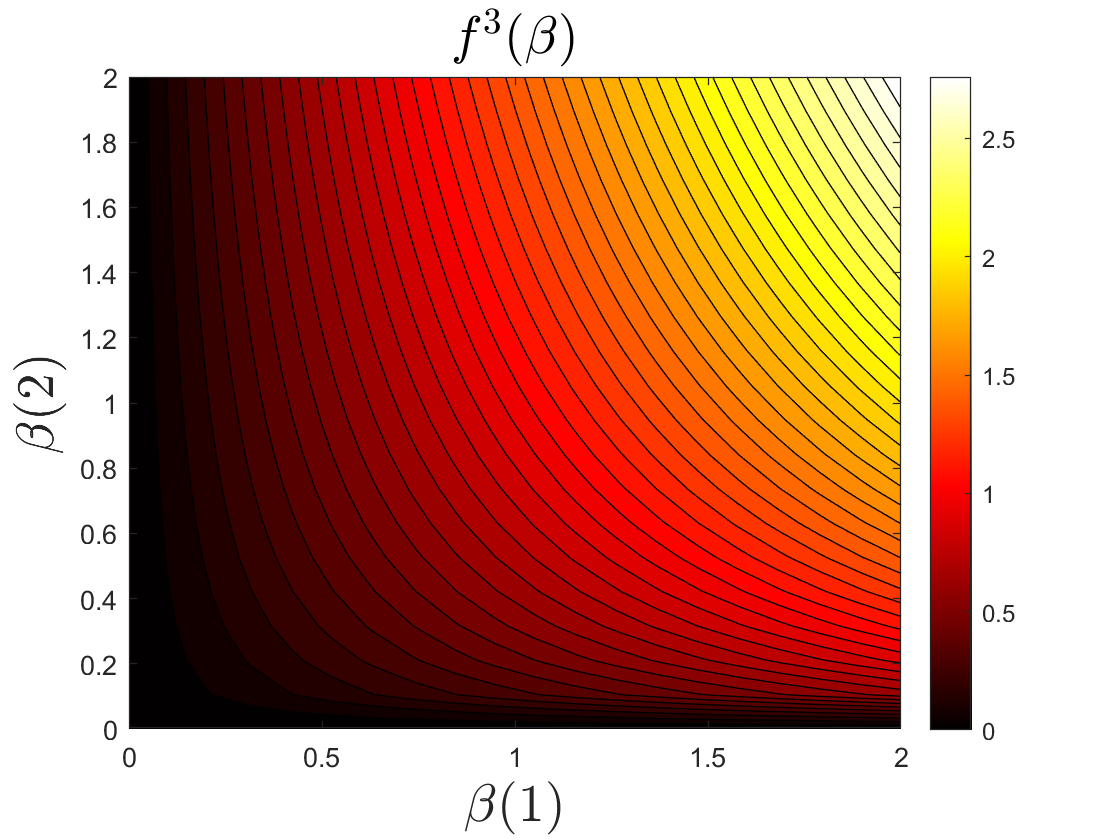}
            \caption[]%
            {{\scriptsize $f^3(\beta) = \sqrt{\beta(1)}\beta(2)$}}    
            \label{fig:T_util3}
        \end{subfigure}
        \hfill
        \begin{subfigure}[]{0.23\textwidth}   
            \centering 
            \includegraphics[width=\textwidth]{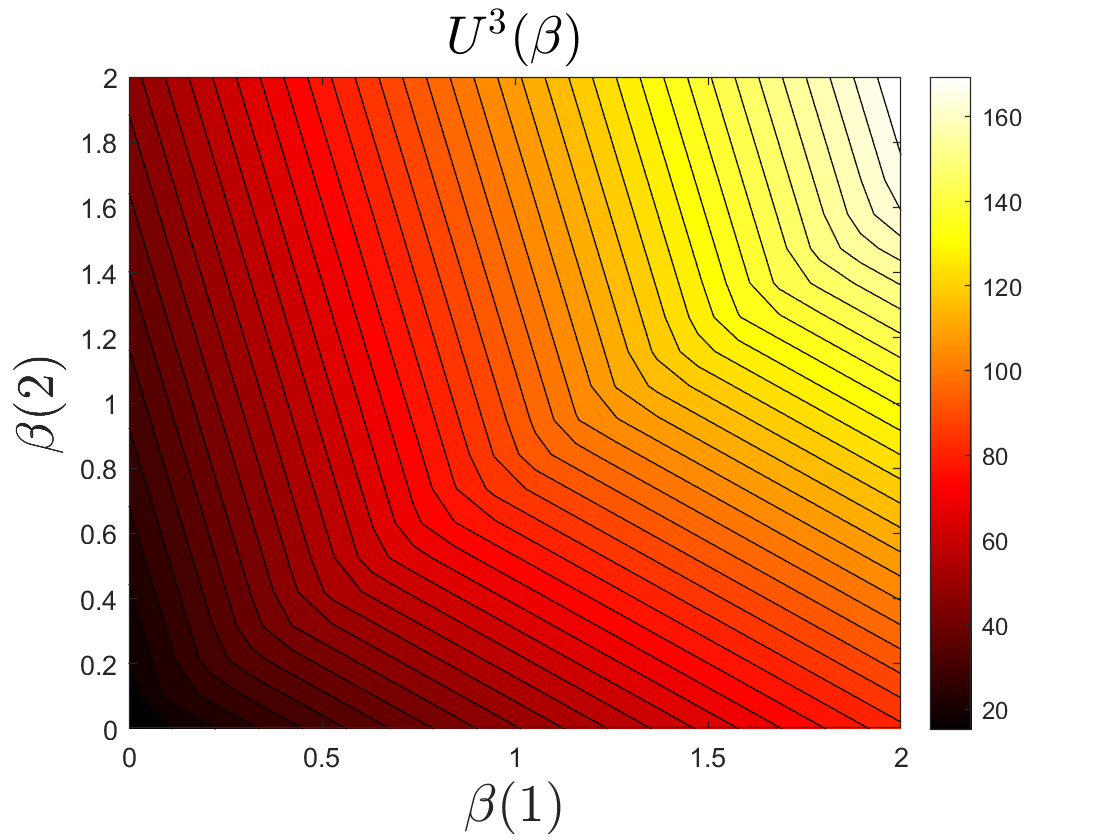}
            \caption[]%
            {{\scriptsize $U^3(\beta)$}}    
            \label{fig:util3}
        \end{subfigure}
        \caption[]
        {\small $f^i(\beta)$ is the true objective function of the $i$'th radar, inducing the responses $\{\hat{\beta}_n^i\}_{n=1}^{10}$. $U^i(\beta)$ is the reconstructed objective function for radar $i$, computed using the noisy dataset $\barbeta$ and \eqref{eq:NumRec}.}
        \label{fig:utilities}
\end{figure}

\section{Conclusion}
\label{sec:conc}
In this work we present a methodology for detecting coordination in a cognitive radar network by observing noisy radar signals. We first present the equivalence between radar network multi-objective optimization (coordination) and a linear program formulation. This allows us to develop a statistical detector, with theoretical performance guarantees, for identifying coordination in the radar network from noisy signals. We present a practical algorithm for implementing the statistical detector and reconstructing functions which approximate the true objective functions in the multi-objective optimization. We present numerical simulations which demonstrate the efficacy of both the statistical detector and the function reconstruction procedure. We note that while we focus on the concrete example of a cognitive radar network, the methodology applies more generally to inverse multi-objective optimization.

\bibliographystyle{IEEEtran}
\bibliography{Bibliography.bib}

\section{Appendix}
\subsection{Lemma 1}
Consider the optimization problem \eqref{eq:moo}. Then 
\begin{equation}
\label{lemeq:mu0}
    \var_t^j > 0 \Rightarrow \sw^j > 0
\end{equation}

\begin{proof}
    Suppose $\sw^j = 0$ and let $\{\var_t^i\}_{i=1}^{\numagents}$ satisfy \[\lconst'(\sum_{i=1}^{\numagents}\var_t^i) \leq p^*\] with $\var_t^j > 0$. Then \[\lconst'(\sum_{i=1}^{\numagents} \var_t^i) = \lconst (\sum_{i\neq j}\var_t^i) + \lconst(\var_t^j) \leq p^* \]
    and \[\sum_{i=1}^{\numagents}\sw^i f^i(\var_t^i) = \sum_{i=1, i\neq j}^{\numagents}\sw^i f^i(\var_t^i)\]
    and since $\lconst > 0$, $\var_t^j > 0$, $\exists \ \delta > 0$ such that
    \[\lconst'(\sum_{i=1, i \neq j}^{\numagents} \var_t^i) \leq p^* - \delta\]

Let \[X_j(\lconst, p^*) := \{\{\var_t^i\}_{i\neq j} : \lconst (\sum_{i=1, i\neq j}^{\numagents} \var_t^i) \leq p^*\} \], and fix some $\var_t^k, k\neq j$. we have that
\[\var_t^k \leq f^{k^{-1}}\left(\frac{1}{\sw^k}(p^* - \delta - \sum_{i \neq k}\sw^i f^i(\var_t^i) )\right)\]
Now take
\[\bar{\var_t}^k =  f^{k^{-1}}\left(\frac{1}{\sw^k}(p^* - \sum_{i \neq k}\sw^i f^i(\var_t^i)) \right)\]
Then, since $f^k$ is monotone increasing, we have 
\begin{align*}
    &\bar{\var_t}^k > \var_t^k, so \\
    &\sum_{i=1}^{\numagents} \sw^i f^i(\var_t^i) < \sum_{i=1,i\neq k}^{\numagents} \sw^i f^i(\var_t^i) + \sw^kf^k(\bar{\var_t}^k)
\end{align*}
and \[\{\var_t^i\}_{i=1,i\neq k}^{\numagents}\cup\{\bar{\var_t}^k\} \in X_j(\lconst,p^*)\] so 
\[\{\var_t^i\}_{i=1}^{\numagents} \notin \arg\max_{\{\gamma^i\}_{i=1}^{\numagents}}\sum_{i=1}^{\numagents} \sw^i f^i(\gamma^i) \ s.t. \ \lconst'(\sum_{i=1}^{\numagents}\gamma^i \leq p^*)\]
and thus by contradiction we have that for any $\sw^j, \var_t^j$ in \eqref{eq:moo}, we have $\sw^j = 0 \Rightarrow \var_t^j = 0$. Note that this directly implies \eqref{lemeq:mu0}.
\end{proof}


\subsection{Proof of Theorem \ref{thm:stat_det}}
\label{pf:stat_det}
\begin{proof}[Proof: 1]
Suppose $H_0$ holds. By Theorem \ref{thm:cherchye1}, $H_0$ is equivalent to \eqref{af_ineq} having a feasible solution. Let $(\bar{u}_t^i,\bar{\lambda}_t^i, t \in [T])_{i=1}^M$ denote a feasible solution to \eqref{af_ineq}. Then substituting $\nrespi = \respi - \noisei_t$, it is apparent that ($\bar{u}_t^i, \bar{\lambda}_t^i, \Phi = \rvtesti$) is feasible. So, clearly the minimizing solution of \eqref{eq:LP} satisfies $\optstati \leq \rvtesti \ \forall i \in [M]$.

Now suppose $\optstati \leq \rvtesti \ \forall i \in [M]$, and let $(\bar{u}_t^i,\bar{\lambda}_t^i)$ denote a feasible solution to \eqref{eq:LP}. Similarly, this implies that \eqref{af_ineq} has a feasible solution, i.e. $H_0$ holds. 
\end{proof}

\begin{proof}[Proof: 2]
From \eqref{eq:H0equiv}, the probability of Type-I error is 
\begin{equation}
\label{eq:T1err1}
    \Popt(H_1 | H_0) = \Prob(\fccdf(\optstat) \leq \gamma \ | \ \bigcap_i \{\optstati \leq \rvtesti\})
\end{equation}
First note that \[\bigcap_i \{\optstati \leq \rvtesti\} \subseteq \{\optstat \leq \rvtest\}\]
and thus \eqref{eq:T1err1} is equivalent to 
\begin{align*}
\begin{split}
\label{eq:T1err2}
    &\Popt(H_1 | H_0) = \\&\Prob(\fccdf(\optstat) \leq \gamma \ | \ \{ \optstat \leq \rvtest\}\bigcap \bigl\{ \bigcap_i \{\optstati \leq \rvtesti\}\bigr\})
\end{split}
\end{align*}
Now if $\optstat = \rvtest$, then since $\fccdf(\rvtest)$ is uniform
in [0,1] we have $\Popt(H_1 | H_0) = \gamma$.
If $\optstat < \rvtest$ then
\begin{align*}
    \begin{split}
&\fccdf(\optstat) \geq \fccdf(\rvtest) \\
&\quad \Rightarrow \Prob(\fccdf(\optstat) \leq \gamma) \leq \Prob(\fccdf(\rvtest) \leq \gamma) \leq \gamma\\
& \quad  \Rightarrow \Popt(H_1 | H_0) \leq \gamma
    \end{split}
\end{align*} 
\end{proof}

\begin{proof}[Proof: 3]
Suppose $\barstati > \optstati \ \forall i \in [M] \Rightarrow \baroptstat := \max_i \barstati > \optstat$. Then we have 
\begin{align*}
    \begin{split}
        &\Prob(\fccdf(\baroptstat) \leq \gamma | \bigcap_i \{\barstati \leq \rvtesti\}) \\
        &\geq P(\fccdf(\optstat) \leq \gamma  | \bigcap_i \{\optstati \leq \rvtesti\} \\
        &\Rightarrow \Prob_{\bar{\Phi}(\barbeta)}(H_1|H_0) \geq \Popt(H_1|H_0) \ \forall \bar{\Phi} \in [\optstat,\rvtest]
    \end{split}
\end{align*}
\end{proof}

\end{document}